\documentclass{llncs}

\usepackage{booktabs}
\usepackage{multirow}
\usepackage[center,footnotesize]{caption}
\usepackage{url}
\usepackage{graphicx}
\usepackage{wrapfig}
\usepackage{amssymb}
\usepackage{subcaption}
\usepackage{appendix}
\usepackage{multicol}
\usepackage[para]{footmisc}
\usepackage{tikz}
\usetikzlibrary{matrix,arrows,decorations.pathmorphing,decorations.pathreplacing,calc, decorations.markings,decorations.text}

\RequirePackage{hyperref}
\RequirePackage[T1]{fontenc}
\RequirePackage{amsmath}
\RequirePackage{xspace}

\newcommand{\method}[1]{{\sc #1}}
\newcommand{\freq}{\ensuremath{\mathit{freq}}}
\newcommand{\surf}{{\sc{surf}}}
\newcommand{\invidx}[0]{\method{invidx}}
\newcommand{\invidxexh}[0]{\method{invidx-e}}
\newcommand{\invidxwand}[0]{\method{invidx-w}}
\newcommand{\idxd}[0]{\method{i-d$^n$}}
\newcommand{\idxdr}[0]{\method{i-d$^n$\^{r}$^n$}}
\newcommand{\idxdIrI}[0]{\method{i-d$^1$\^{r}$^1$}}

\newcommand{\col}{\ensuremath{\mathcal{D}}\xspace} 
\newcommand{\docs}{\ensuremath{N}\xspace}          
\newcommand{\ccol}{\ensuremath{\mathcal{C}}\xspace} 

\newcommand{\sentinel}{0}
\newcommand{\tokens}{\ensuremath{n}}
\newcommand{\docd}{\ensuremath{d}}
\newcommand{\qryq}{\ensuremath{Q}}
\newcommand{\score}{\ensuremath{\mathcal{S}}\xspace} 

\newcommand{\expandnode}{\method{expand}}
\newcommand{\termt}{\ensuremath{q}}
\newcommand{\FDq}{\ensuremath{F_{\col,\termt}}}
\newcommand{\FDvq}{\ensuremath{F_{\col_{v},\termt}}}
\newcommand{\fDq}{\ensuremath{f_{\col,\termt}}}
\newcommand{\fDvq}{\ensuremath{f_{\col_{v},\termt}}}
\newcommand{\deltafF}{\ensuremath{\delta_{\col_{v},\termt}}}
\newcommand{\colC}{\ensuremath{\mathcal{C}}}
\newcommand{\fdq}{\ensuremath{f_{\docd,\termt}}}
\newcommand{\wdt}{\ensuremath{w_{\docd,\termt}}}
\newcommand{\fqt}{\ensuremath{f_{Q,\termt}}}
\newcommand{\wqt}{\ensuremath{w_{Q,\termt}}}

\newcommand{\DARRAY}{\ensuremath{D}}     
\newcommand{\REP}{\ensuremath{R}}    
\newcommand{\REPHAT}{\ensuremath{\hat{R}}}    
\newcommand{\REPHATN}{\ensuremath{\hat{R}^n}}    
\newcommand{\SUF}{\ensuremath{SA}}   
\newcommand{\DOCLEN}{\ensuremath{L}}

\newcommand{\algname}[1]{\mbox{#1}}

\newcommand{\GREEDY}{\algname{GREEDY}}

\newcommand{\structname}[1]{\mbox{#1}}
\newcommand{\CSA}{\structname{CSA}}

\newcommand{\WT}{\structname{WT}}
\newcommand{\WTX}{\structname{WT-X}}
\newcommand{\WTDT}{\structname{WT-D}} 
\newcommand{\WTD}{\structname{WT-D$^n$}}
\newcommand{\WTDX}{\structname{WT-D$^{\ell}$}}
\newcommand{\WTREPT}{\structname{WT-\REP}} 

\newcommand{\WTREPTHAT}{\structname{WT-\REPHAT}} 
\newcommand{\WTREPHAT}{\structname{WT-\REPHAT$^n$}}
\newcommand{\WTREPHATI}{\structname{WT-\REPHAT$^1$}}
\newcommand{\WTREPHATX}{\structname{WT-\REPHAT$^{\ell}$}}
\newcommand{\WTDI}{\structname{WT-D$^1$}}

\newcommand{\DFSADA}{\structname{DF}}

\newcommand{\Order}[1]{\ensuremath{\mathcal{O}(#1)}}

\newcommand{\gb}[1]{{\mbox{$#1$~GiB}}}

\newcommand{\bm}{\mbox{\scriptsize BM25}\xspace}
\newcommand{\lmds}{\mbox{\scriptsize LMDS}\xspace}
\newcommand{\tfxidf}{\mbox{\scriptsize \textsc{TF}$\times$\textsc{IDF}}\xspace}

\DeclareMathOperator{\df}{\mbox{\scriptsize \textsc{df}}\xspace}

\DeclareMathOperator{\LM}{LM}

\newcommand{\sdsl}{{\sc{sdsl}}}

\newcounter{todocount}
\newcommand{\sourcedir}{http://people.eng.unimelb.edu.au/sgog/DCC2014}
\newcommand{\dsl}[1]{%
\sourcedir/info/#1.html%
}%
\newcommand{\sdslgitinc}[0]{https://github.com/simongog/sdsl-lite/tree/master/include/sdsl}

\newcommand{\govtwo}[0]{\textsc{Gov2}}
\newcommand{\clueweb}[0]{\textsc{ClueWeb09}}
\newcommand{\trecqryA}[0]{\textsc{Trec2005}}
\newcommand{\trecqryB}[0]{\textsc{Trec2006}}

\newcommand{\rankor}{\mbox{\scriptsize Ranked-\textsc{OR}}\xspace}
\newcommand{\rankand}{\mbox{\scriptsize Ranked-\textsc{AND}}\xspace}

\tikzstyle{sleaf} = [inner sep=0mm,minimum size=3mm,rectangle,fill=brown!20,draw=brown!50,thick]
\tikzstyle{sdoc} = [inner sep=0mm,minimum size=3mm,rectangle,fill=gray!40,draw=gray!80,thick]
\tikzstyle{sinner} = [inner sep=0mm,minimum size=4mm,circle,fill=blue!20,draw=blue!50,thick]
\tikzstyle{svirtual} = [inner sep=0mm,minimum size=4mm,circle,fill=white!20,draw=blue!50,thick]
\tikzstyle{stext} = [color=brown]
\tikzstyle{scstlabel} = [sloped,midway,font=\scriptsize,above]

\def\vpos{ 0,1,2,3,4,5,6,7,8,9,10,11,12,13,14,15,16,17,18,19,20,21,22,23,24,25,26}
\def\lev{{2,0.3,2.5,1.5,4,2.5,4,0,4.3, 3.5, 5.8, 4.3, 5.8, 2.5, 5, 3.5, 5, 1.5, 3, 0.3, 5.3, 4, 5.3, 3, 4.3, 1.5, 3}}
\def\vtyp{{0,2,0,1,0,1,0,1,0,1, 0, 1, 0, 2, 0, 1, 0, 1, 0, 2, 0, 1, 0, 1, 0, 1, 0}}
\def\myedge{1,7,3,1,5,3,5,7,9,13,11,9,11,17,15,13,15,19,17,7,21,23,21,25,23,19,25}
\def\edgela{%
{"\$","","\$","\#","\$","O","LA..","","\$","\#","LA..","O","O..","","\#..","LA",%
 "LA..","LA","O..","","\$","\#","O..","LA","LA..","O","O.."}%
}
\def\sa{{13,12,3,8,11,2,7,6,5,0,10,1,4,9}}
\newcounter{ii}

\def\daperm{0,2,1,3,2,1,3,3,3,1,2,1,3,2}
\def\dapermZ{0,1,1,1,1}
\def\dapermO{2,3,2,3,3,3,2,3,2}
\def\dapermZZ{0}
\def\dapermZO{1,1,1,1}
\def\dapermOZ{2,2,2,2}
\def\dapermOO{3,3,3,3,3}
\newdimen\myx

\newcommand{\drawWTnode}[4]{
    \setcounter{ii}{-1}
    \foreach \d [count=\ii from 0] in #1{
        \node[color=gray, inner sep=0, minimum width=0.5cm,minimum height=0.3cm,font=\scriptsize] (#2_\ii) at (0.6*\ii cm, 0cm) {\d};
        \ifnum #3=#4 \else
            \node[color=black, inner sep=0, minimum width=0.5cm,minimum height=0.3cm,below] (#2_b\ii) at (#2_\ii.south) {}; 
        \fi
        \stepcounter{ii}
    }
    \ifnum #3=#4
        \path %
            let \p1 = ($(#2_\theii.east)-(#2_0.west)$) %
               ,\n1 = {(veclen(\x1,\y1))}
            in %
            node[inner sep=0,minimum width=\n1, minimum height=0.3cm,draw=white,thick,rounded corners=5pt,
                below right]
            (#2) at (#2_0.north west) {};
    \else
        \path %
            let \p1 = ($(#2_b\theii.east)-(#2_b0.west)$) %
               ,\n1 = {(veclen(\x1,\y1))}
            in %
            node[inner sep=0,minimum width=\n1, minimum height=0.3cm,draw=gray!80,thick,rounded corners=5pt,
                fill=gray!20,below right]
            (#2) at (#2_b0.north west) {};
        \foreach \d [count=\ii from 0] in #1{
            \pgfmathparse{int(2*(\d-#3) > (#4-#3)))}\let\mybit\pgfmathresult
            \node[color=black, inner sep=0, minimum width=0.5cm,minimum height=0.3cm,below] (#2_bit\ii) at (#2_\ii.south) {\mybit}; 
        }
    \fi
}

\newcommand{\printInterval}[4]{
    \coordinate (c1) at ($(#1_#2.north west)+(0,1.5ex)$);
    \coordinate (c2) at ($(#1_bit#2.south west)+(0,-0.5ex)$);
    \coordinate (c3) at ($(#1_#3.north east)+(0,1.5ex)$);
    \coordinate (c4) at ($(#1_bit#3.south east)+(0,-0.5ex)$);

    \path[fill opacity=0.3,fill=brown!80]  (c1.north west) rectangle (c4.south east);
    \draw[brown!80,thick]  (c1) -- (c2);
    \draw[brown!80,thick]  (c1) -- ($(c1)+(1.5ex,0)$);
    \node[below right,inner sep=0.3mm, font=\tiny,color=brown!80] at (c1) {#2};
    \draw[brown!80,thick]  (c3) -- (c4);
    \draw[brown!80,thick]  (c3) -- ($(c3)+(-1.5ex,0)$);
    \node[below left,inner sep=0.3mm, font=\tiny,color=brown!80] at (c3) {#3};
    \node[below,color=brown!80,font=\tiny,inner sep=0.3mm] at ($(c1)+0.5*(c3)-0.5*(c1)$) {#4};
}

\tikzstyle{swtlab} = [sloped,above,black,inner sep=0.5mm,font=\scriptsize]

\newcommand{\printWTD}{
    \begin{scope}
        \drawWTnode{\daperm}{v0}{0}{3} 
    \end{scope}
    \begin{scope}[yshift=-1.25cm,xshift=-2.5mm]
        \drawWTnode{\dapermZ}{v1}{0}{1} 
    \end{scope}
    \begin{scope}[yshift=-1.25cm,xshift=3.25cm]
        \drawWTnode{\dapermO}{v2}{2}{3} 
    \end{scope}
    \begin{scope}[yshift=-2.5cm,xshift=-3.5mm]
        \drawWTnode{\dapermZZ}{v3}{0}{0} 
    \end{scope}
    \begin{scope}[yshift=-2.5cm,xshift=5mm]
        \drawWTnode{\dapermZO}{v4}{1}{1} 
    \end{scope}
    \begin{scope}[yshift=-2.5cm,xshift=31.5mm]
        \drawWTnode{\dapermOZ}{v5}{2}{2} 
    \end{scope}
    \begin{scope}[yshift=-2.5cm,xshift=59mm]
        \drawWTnode{\dapermOO}{v6}{3}{3} 
    \end{scope}
    \draw[gray!80, thick, ->,>=stealth'] (v0.south) to node[swtlab]{0} (v1.north);
    \draw[gray!80, thick, ->,>=stealth'] (v0.south) to node[swtlab]{1} (v2.north);
    \draw[gray!80, thick, ->,>=stealth'] (v1.south) to node[swtlab]{0} (v3.north);
    \draw[gray!80, thick, ->,>=stealth'] (v1.south) to node[swtlab]{1} (v4.north);
    \draw[gray!80, thick, ->,>=stealth'] (v2.south) to node[swtlab]{0} (v5.north);
    \draw[gray!80, thick, ->,>=stealth'] (v2.south) to node[swtlab]{1} (v6.north);

    \printInterval{v0}{4}{9}{LA}
    \printInterval{v1}{2}{3}{LA}
    \printInterval{v2}{2}{5}{LA}
}

\def\reps{{7/{1,2,3}/0},{13/{3}/0},{15/{3}/1.0},{17/{1}/-1.0},{19/{1,2,3}/0},{25/{2}/0}}
\def\maptype{{"sleaf","sinner","svirtual"}}
\newcommand{\printsadatree}[1]{
    \setcounter{ii}{0}
    \foreach \v [count=\ii from 1, count=\lii from 0] in \vpos{
        \pgfmathparse{\vtyp[\v]}\let\ttt\pgfmathresult
        \pgfmathparse{\maptype[\vtyp[\v]]}\let\res\pgfmathresult
        \pgfmathparse{\lev[\v]}\let\level\pgfmathresult
        \pgfmathparse{ifthenelse((#1>0) || (\ttt>0),1,0)}\let\dodraw\pgfmathresult
        \ifnum \dodraw>0
            \node[\res] (v\v) at (0.4*\lii cm,-0.6*\level cm) {};
            \ifnum \ttt>0
                \stepcounter{ii}
                \node[blue] at (v\v) {{\scriptsize $v_{\theii}$}};
            \fi
         \fi
    }
    \foreach \v/\list/\xadj in \reps{
        \node[above,xshift=\xadj ex,inner sep=1mm] at (v\v.north) {$\{\textcolor{black}{\list}\}$};
    }
    \setcounter{ii}{0}
    \foreach \dest [count=\src from 0] in \myedge{
        \pgfmathparse{\vtyp[\src]}\let\ttt\pgfmathresult
        \pgfmathparse{\edgela[\src]}\let\lll\pgfmathresult
        \ifnum \src=\dest \else
            \pgfmathparse{ifthenelse(\lev[\src]>\lev[\dest],\src,\dest)}\let\w\pgfmathresult
            \pgfmathparse{ifthenelse(\lev[\src]>\lev[\dest],\dest,\src)}\let\v\pgfmathresult
            \pgfmathparse{ifthenelse(\src<\dest,"mirror","")}\let\mir\pgfmathresult
            \pgfmathparse{ifthenelse((#1>0) || (\ttt>0),1,0)}\let\dodraw\pgfmathresult
            \ifnum \dodraw>0
                \draw[black,->] (v\v) -- (v\w); 
                \path[text=brown, postaction={decorate,%
                      decoration={\mir,text along path, text align=center,%
                         raise=0.5ex,
                         text={\textcolor{brown}{\scriptsize\lll}}{}}}] (v\v) -- (v\w); 
            \fi
        \fi
    }
    \foreach \v in \vpos {
        \pgfmathparse{\vtyp[\v]}\let\type\pgfmathresult
        \ifnum \type = 0
            \pgfmathparse{\sa[\theii]}\let\suf\pgfmathresult
            \pgfmathparse{{\daperm}[\theii]}\let\ddd\pgfmathresult
            \node[color=brown] at (v\v) {\suf};
            \node[below,color=gray] at (v\v.south) {\ddd};
            \stepcounter{ii} 
        \fi
    }

}

\title{Compact 
Indexes for Flexible Top-$k$ Retrieval}

\author{Simon Gog
\and Matthias Petri
}

\institute{%
Department of Computing and Information Systems\\
The University of Melbourne, Parkville VIC 3010, Australia
}

\begin{document}

\maketitle

\begin{abstract}
We engineer a self-index based retrieval system capable
of rank-safe evaluation of top-$k$ queries.
The framework generalizes the \GREEDY\ approach of Culpepper et 
al.~(ESA 2010) to handle multi-term queries, including over phrases.
We propose two techniques which significantly reduce
the ranking time for a wide range of popular Information Retrieval
(IR) relevance measures, such as $\tfxidf$ and $\bm$.
First, we reorder elements in the document array according
to document weight. Second, we introduce the repetition
array, which generalizes Sadakane's (JDA 2007)
document frequency structure to document subsets.
Combining document and repetition array, 
we achieve attractive functionality-space trade-offs.
We provide an extensive evaluation of
our system on terabyte-sized IR collections.

\end{abstract}

\section{Introduction}

Calculating the $k$ most relevant documents for a multi-term query $\qryq$ 
against a set of documents $\col$ is a fundamental problem -- the top-$k$ document 
retrieval problem -- in Information Retrieval (IR).
The relevance of a document $\docd$ to $\qryq$ is determined by evaluating
a similarity function $\score$ (e.g. \tfxidf\ or \bm). Naive exhaustive processing
evaluates $\score$ for each document $d$ in $\col$ and generates a full list of 
scores.
The top-$k$ documents in the list
are then reported. Algorithms which guarantee production of the
same top-$k$ results list as the exhaustive process are called \emph{rank-safe}.

The \emph{inverted index} is a highly-engineered data structure designed to solve this
problem. The index stores, for each unique term in $\col$, the list of documents $d_i$
containing that term. Queries are answered by processing the lists of all of the 
query terms. Advanced query processing schemes \cite{bch+03-cikm} process
lists only partially while remaining rank-safe. However, additional
work during construction time is required to avoid scoring non-relevant documents at
query time. Techniques used to speed up query processing include sorting 
lists in decreasing score order, or pre-storing score upper bounds 
for sets of documents which can then safely be skipped during query processing. 
These pre-processing steps introduce a dependency between the similarity
measure and the stored index. Changing the scoring function requires at
least partial reconstruction of the underlying index, which in turn reduces
the flexibility of the retrieval system.

Another family of retrieval systems is based on self-indexes~\cite{NavACMcs14}.
These systems support functionality not easily provided by inverted indexes,
such as efficient phrase search, and direct text extraction. 
Systems capable of single-term top-$k$ queries have been proposed~\cite{nn-soda12}
and have proven to work well in practice~\cite{kn-dcc13}.
Generalizing and extending these structures to support
multi-term queries and more complex similarity functions is essential to the
adaption of self-indexes in the context of IR. However, currently only simple heuristics 
which cannot provide rank-safe query processing exist~\cite{cps-sigir12}.

\emph{Our Contributions.} 
We propose, to the best of our knowledge, the first self-index based
retrieval framework capable of rank-safe evaluation of top-$k$ queries.
In addition to the functionality of self-indexes (such as text extraction and
phrase queries) it can process multi-term queries using complex IR relevance
measures on terabyte scale IR collections.
It is based on a generalization of the \GREEDY\ approach of 
Culpepper et al.~\cite{CUL:NAV:PUG:TUR:2010}. We suggest two
techniques to decrease the number of evaluated nodes in the
\GREEDY\ approach. The first is reordering of documents
according to their length (or other suitable weight), the
second is a new structure called the \emph{repetition array}, \REP.
The latter is derived from Sadakane's \cite{sada07jda} 
document frequency structure, and is used to
calculate the document frequency for subsets of documents.
We further show that it is sufficient to use only \REP\ and a subset
of the document array if query terms, which can also be phrases, are 
length-restricted.
Finally, we explore the properties of our proposal on two terabyte-scale IR 
collections. This is, to our knowledge, three orders of
magnitudes larger than in previous self-index based experiments.
Our source code and experimental setup is publicly available.

\emph{Paper outline.} 
In Section \ref{sec-notation} we introduce notation, a formal problem
definition, and examples of similarity measures. Section~\ref{sec-toolbox}
reviews essential data structures. Sections $\ref{sec-previous-work}$ and
\ref{sec-improvements} revisit, generalize, and improve the
\GREEDY\ method.
Finally we evaluate our proposals
in Section~\ref{sec-experiments} and conclude in Section~\ref{sec-conclusion}.

\section{Notation and Problem Definition} \label{sec-notation}

Let $\col'=\{\docd_1,\ldots,\docd_{\docs\!-\!1}\}$ be a collection of $\docs\!-\!1$
documents. Each $d_i$ is a string over an alphabet
\footnote{Note: In Information Retrieval $\Sigma$\ is usually a word alphabet
and in String Processing a character alphabet.}
$\Sigma'=[2,\sigma]$ 
and is terminated by the sentinel symbol `$1$', also represent as `\#'.
Adding the one-symbol document $\docd_{0}=\sentinel$ results in a collection $\col$
of $\docs$ documents. The concatenation 
$\ccol=\docd_{\pi(0)}\docd_{\pi(1)}\ldots\docd_{\pi(\docs\!-\!1)}$ 
is a string over $\Sigma=[0,\sigma]$, where $\pi$ is a permutation
of $[0,\docs\!-\!1]$ with $\pi(\docs\!-\!1)=0$.
We denote the length of a document $\docd_i$ with $|\docd_i|=n_{d_i}$, and
$|\ccol|=n$. See Fig.~\ref{fig-col} for a running example. 
In the ``bag of words'' model a query $\qryq=\{q_0,q_1,\ldots,q_{m-1}\}$ is an 
unordered set of length $m$. Each element $q_i$ is either
a \emph{term} (chosen from $\Sigma'$) or a \emph{phrase} (chosen
from ${\Sigma'}^p$ for $p>1$). We can now define our problem.

\newcounter{mydoccnt}
\def\myintcol{2,3,2,1,3,2,2,2,1,3,3,2,1,0}
\def\alphabet{{"\noexpand\$","\noexpand\#","LA","O"}}
\def\mydocs{0/3,4/8,9/12,13/13}
\def\mydocperm{{1,3,2,0}}
\tikzstyle{mynode} = [inner sep=0mm,minimum height=3ex,minimum width=4ex]
\setcounter{mydoccnt}{0}
\begin{figure}[ht]
\centering
\begin{tikzpicture}
    \foreach \intword [count=\ii from 1, count=\lii from 0] in \myintcol{
	    \node[mynode,stext] (col\lii) at (0.7*\lii cm,0) {%
            \pgfmathparse{\alphabet[\intword]}\pgfmathresult%
        };
        \node[above,mynode] (idx\lii) at (col\lii.north) {\textcolor{gray}{\lii}};
        \node[below,mynode] (intcol\lii) at (col\lii.south) {\intword};
        \ifnum \intword=1%
            \stepcounter{mydoccnt}
        \fi
    }
    \node[left] at (idx0.west) {\textcolor{gray}{$i=$}};
    \node[left] at (col0.west) {$\ccol^{word}=$};
    \node[left] at (intcol0.west) {$\ccol=$};
    \foreach \sp/\ep [count=\ii from 0] in \mydocs{
        \pgfmathparse{int(\mydocperm[\ii])}\let\pii\pgfmathresult
        \draw [decorate,decoration={brace,amplitude=5pt}]
               (intcol\ep.south east) -- (intcol\sp.south west) node [black,midway,yshift=-10pt] {
                   \footnotesize $d_\pii$};
    }
\end{tikzpicture}
\vspace{-1ex}
\caption{$\ccol$ is the concatenation of a document 
collection $\col$ for $\pi=[1,3,2,0]$. 
We use both words (as in $\ccol^{word}$) or integer identifiers
(as in $\ccol$) to refer to document tokens. 
 \label{fig-col}}
\end{figure}
\emph{Top-$k$ document retrieval problem.}
Given a collection \col, a query $\qryq$ of length $m$, and a similarity
measure $\score:\col\times{\Sigma'}^m\rightarrow\mathbb{R}$.
Calculate the top-$k$ documents of \col\ with regard to $\qryq$ and
$\score$. That is a sorted list of document identifiers
$\mathrm{T}=\{\tau_0,\ldots,\tau_{k-1}\}$, with
$\score(\docd_{\tau_i},\qryq)\geq \score(\docd_{\tau_{i+1}},\qryq)$ for $0\leq i<k$ and
$\score(\docd_{\tau_{k-1}},\qryq)\geq \score(\docd_{j},\qryq)$ for  $j\not\in\mathrm{T}$.

A basic similarity measure used in many self-index based document
retrieval systems (see \cite{NavACMcs14}),
is the \emph{frequency measure} $\score^{freq}$.
It scores $d$ by accumulating the \emph{term frequency} of each term.
Term frequency $\fdq$ is defined as the number of occurrences of
term $\termt$ in $\docd$;
e.g. 
$f_{\docd_1,\mbox{\footnotesize\textcolor{brown}{LA}}}=2$ in Fig. \ref{fig-col}.
In IR, more complex $\tfxidf$ measures also include two additional factors.
The first is the inverse of the \emph{document frequency} ($\df$), which
is the number of documents in $\col$ that contain $q$, defined $\FDq$;
e.g. $F_{\col,\mbox{\footnotesize\textcolor{brown}{LA}}}\!=\!3$. 
The second is the length of the document $n_d$. Due to space limitations, we only
present 
the popular Okapi \bm function: 
\begin{equation}
    S_{\qryq,\docd}^{\bm}=%
        \sum_{\termt\in \qryq}%
            \underbrace{
            \frac{(k_1+1)\fdq}{k_1\left(1-b+b\frac{n_{\docd}}{n_{\mbox{avg}}}\right)+\fdq}    %
            }_{=\wdt} \cdot %
        \underbrace{\fqt\cdot\ln\left(\frac{\docs-\FDq+0.5}{\FDq+0.5}\right)}_{=\wqt}
\label{eq-bm25}        
\end{equation}
where $n_{\mbox{avg}}$ is the mean document length, and $\wdt$ and $\wqt$ refer
to components that we address shortly.
Parameters $k_1$ and $b$ are commonly set to $1.2$ and $0.75$ respectively.
Note that the $\wqt$ part is negative for $\FDq>\frac{\docs}{2}$.
To avoid negative scores, real-world systems, such as Vigna's
MG4J \cite{BoVTREC2005} search engine, set $\wqt$ to a small
positive value ($10^{-6}$), in this case.
We refer to Zobel and Moffat \cite{zm06compsurv} for a survey on
IR similarity measures including \tfxidf, \bm, and \lmds.

\section{Data Structure Toolbox} \label{sec-toolbox}
We briefly described the two most important building blocks
of our systems, and refer the reader to Navarro's survey \cite{NavACMcs14} 
and references therein for detailed information.
A \emph{wavelet tree} (WT) of a sequence $X[0,n\!-\!1]$ over alphabet $\Sigma[0,\sigma\!-\!1]$
is a perfectly balanced binary tree of height $h=\lceil\log\sigma\rceil$, referred
to as $\WTX$. The $i$-th node of level $\ell\in [0,h\!-\!1]$ is associated with symbols
$c$ such that $\lceil c/2^{h-1-\ell}\rceil=i$. Node $v$, corresponding to 
symbols $\Sigma_v=[c_b,c_e] \subseteq [0,\sigma-1]$, represent the subsequence
$X_v$ of $X$ filtered by symbols in $\Sigma_v$.
Fig. \ref{fig-WT} depicts an example. Only the bitvector which
indicates if an element will move to the left or right subtree is stored at
each node;
that is, $\WTX$ is stored in $n\lceil\log\sigma\rceil$ bits.
Using only sub-linear extra space it is possible to
efficiently navigate the tree.
Let $v$ be the $i$-th node on level $\ell<h\!-\!1$, then method
$\expandnode(v)$ returns in constant time a node
pair $\langle u,w\rangle$, where $u$ is the $(2\cdot i)$-th and
$w$ the $(2\cdot i+1)$-th node on level $\ell+1$.
A range $[l,r]\subseteq [0,n\!-\!1]$ in $X$ can be
mapped to range $[l,r]_{v}$ in
node $v$ such that the sequence $X_v[l,r]_{v}$ 
represents $X[l,r]$ filtered by $\Sigma_v$.
Operation $\expandnode(v,[l,r]_{v})$ then returns in constant time
a pair of ranges $\langle[l,r]_{u},[l,r]_{w}\rangle$ 
such that the sequence $X_u[l,r]_{u}$ 
(resp. $X_w[l,r]_{w}$) represents $X[l,r]$ filtered 
by $\Sigma_u$ (resp. $\Sigma_w$).
Fig.~\ref{fig-WT} provides an example.
\begin{figure}
\centering
\begin{tikzpicture}
\printWTD
\end{tikzpicture}
\vspace{-1ex}
\caption{Wavelet tree over document array \DARRAY.
 Method $\expandnode(v_{root},[4,9])$ maps  
range $[4,9]$ (locus of \textcolor{brown}{LA})
to range $[2,3]$ in the left and range $[2,5]$ in the right child.
\label{fig-WT}}
\end{figure}
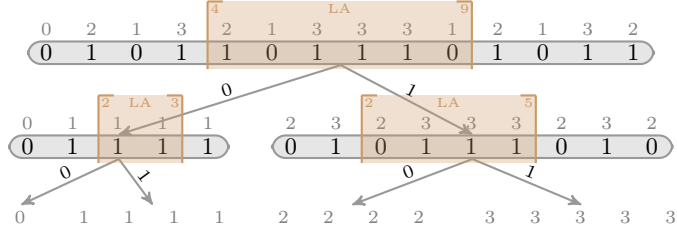

The \emph{binary suffix tree} (BST) of string $X[0,n\!-\!1]$ is the compact
binary trie of all suffixes of $X$. For each path $p$ from the root to a
leaf, the concatenation of the edge labels of $p$, corresponds
to a suffix. The children of a node are ordered lexicographically
by their edge labels. Each leaf is labeled with the starting
position of its suffix in $X$. Read from left to right, the leaves 
form the \emph{suffix array} (\SUF), which is a permutation of
$[0,n-1]$ such that such that 
$X[\SUF[i],n\!-\!1]<_{lex} X[\SUF[i\!+\!1],n\!-\!1]$ for all $0\leq i < n\!-\!1$.
We refer to Fig. \ref{fig-df}\ for an example. 
Compressed versions of \SUF\ and ST -- the 
compressed \SUF\ (CSA) and compressed ST (CST) -- use space essentially
equivalent to that of the compressed input, while efficiently supporting
the same operations.
For example, given a pattern $P$ of length $m$, the range $[l,r]$
in \SUF\ containing all suffixes start with $P$ 
or the corresponding node, that is the \emph{locus} of $P$, in the BST
is calculated in $\Order{m\log\sigma}$.

\newcounter{lastnode}
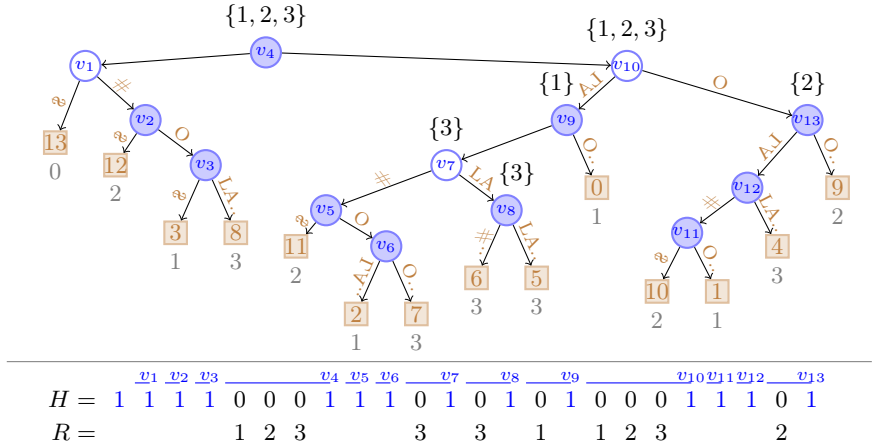
\begin{figure}
\centering
\begin{tikzpicture}
\printsadatree{1}
\end{tikzpicture}
\def\Harray{1,1,1,0,0,0,1,1,1,0,1,0,1,0,1,0,0,0,1,1,1,0,1}
\def\Hstyle{{"black","blue"}}
\def\reps{{1,2,3,3,3,1,1,2,3,2}}
\setcounter{ii}{0}
\setcounter{lastnode}{0}
\setcounter{mydoccnt}{0}
\begin{tikzpicture}
    \draw[thin,gray] (-1.5cm,0.5cm) -- (10cm,0.5cm);
    \foreach \h [count=\lii from 0, count=\ii from 1] in \Harray{
        \pgfmathparse{\Hstyle[\h]}\let\mystyle\pgfmathresult
        \node[\mystyle] (h\ii) at (0.4*\ii cm,0cm) {\h};
        \ifnum \h=0
            \pgfmathparse{\reps[\theii]}\let\ddd\pgfmathresult
            \node[below] (R\theii) at (h\ii.south) {\ddd};
            \stepcounter{ii}
        \else
            \stepcounter{mydoccnt}
            \node[above,inner sep=0mm,color=blue] at (h\ii.north) {{\scriptsize $v_{\themydoccnt}$}};
            \pgfmathparse{int(\thelastnode+1)}\let\iii\pgfmathresult
                \typeout{...\thelastnode}
            \draw[blue,thin] (h\ii.north) -- (h\iii.north west);
            \setcounter{lastnode}{\ii}
        \fi
    }
    \node[blue,dotted] (h0) at (0cm,0cm) {1};
    \node[below] (R0) at (h0.south) {\phantom{1}};
    \node[left] at (h0.west) {$H=$};
    \node[left] at (R0.west) {$R=$};
\end{tikzpicture}
\vspace{-1ex}
\caption{Top: BST of the example in Fig. \ref{fig-col}. 
The leaves form \SUF, the gray numbers below form \DARRAY. 
Bottom: Bitvector $H[0,2n\!-N\!-1]$ and repetition array $\REP$.
\label{fig-df}}
\end{figure}

\section{Revisiting and generalizing the \GREEDY\ framework} \label{sec-previous-work}

The \GREEDY\ framework of Culpepper et al. \cite{CUL:NAV:PUG:TUR:2010} 
consists of two parts: a CSA built over $\ccol$, and a WT over the 
\emph{document array} $\DARRAY[0..\docs-1]$; with
document array entry $\DARRAY[i]$ specifying the document in which
suffix $\SUF[i]$ starts. 
The grey numbers below each $\SUF[i]$ value in Fig.~\ref{fig-df} 
correspond to $\DARRAY[i]$.
A top-$k$ query using $\score^{\freq}$ with $m=1$ is answered as follows. 
For term $\termt=q_0$ the CSA returns a range 
$[l,r]$, such that all suffixes in $\SUF[l,r]$ are prefixed by 
$\termt$. Note that the size of the range corresponds 
to $\fDq$, the number of occurrences of $\termt$ in $\col$.
In $\WTDT$ the alphabet $\Sigma_v$ of each node represents
a subset $\col_v\subseteq \col$ of documents of $\col$;
and the size of the mapped interval $[l,r]_{v}$ equals
$f_{\col_v,\termt}$, the number of occurrences of $\termt$ in the sub-collection $\col_v$.
Each leaf $v$ in $\WTDT$ corresponds to a document $\docd$ in $\col$, such that
the size of $[l,r]_{v}$ equals term frequency $\fdq$.

To calculate the documents with maximal $\fdq$, i.e. maximizing $\score^{\freq}_{\termt,\docd}$, 
a max priority queue stores $\langle v,[l,r]_{v}\rangle$-tuples 
sorted according to interval size.
Initially, $\WTDT$'s root node and $[l,r]$ is enqueued.
The following process is then repeated until $k$ documents are reported
or the queue is empty: dequeue the top element $\langle v,[l,r]_{v}\rangle$.
If $v$ is a leaf, the corresponding document is reported.
Otherwise the node and interval is expanded and the two tuples
$\langle u,[l,r]_{u}\rangle$  and $\langle w,[l,r]_{w}\rangle$ 
 containing the expanded ranges are enqueued.

This iterative process returns the correct result
if the interval size $\fDq$ at a parent is never smaller
than that of a child ($f_{\col_u,q}$ or $f_{\col_w,q}$).
Note that the interval size $\fDq$ is never smaller than the 
maximum $\fdq$ value in the subtree.
We note that in general the algorithm is correct, if (1) the score estimate
$s_v$ at any node $v$ is larger than or equal to the maximum document score 
in $v$'s subtree and (2) the score estimates $s_u$ and $s_w$ 
of the children of $v$ are not larger than $s_v$.

For a wide range
of similarity measures (including $\tfxidf$, \bm, and \lmds)
theses two condition can be established by calculating $s_v$ as follows:
first, all document-independent parts, such as query weight 
$\wqt$ are determined. Then $n_d$ is estimated with 
the smallest document length $n_{\mbox{min}}$ in $\col$ if $v$
is a inner node.  Last, the maximal term frequency $f_{\docd,\termt}$ of each term $\termt_i$ is
set to $f_{\col,\termt_i}$, the size of interval $[l_i,r_i]_{v}$.
Since each interval size
is non-increasing when traversing down $\WTDT$ the algorithm
is correct, but not necessarily very efficient. 

The queue stores states of the form
$\langle s_v, v, \{[l_0,r_0]_{v},\ldots,[l_{m-1},r_{m-1}]_{v}\}\rangle$
sorted by $s_v$ in the multi-term version. 
Processing a state takes $\Order{m}$ time as $m$ intervals
are expanded.

\section{Improving Score Estimation} \label{sec-improvements}

The query time of \GREEDY\ is dependent on the time
to process a state in the \WT\ and on the number of
states processed. The latter is determined by the
quality of the score estimations. 

\subsection{Length Estimation by Document Relabelling}
We first improve document length estimation in $\col_v$ 
by replacing the collection-wide value $n_{min}$ by the smallest
document length $n_{\tilde\docd}$ in the sub-collection $\col_v$.
The computation of $n_{\tilde\docd}$ can be performed in
constant time if the document identifiers are assigned according
to the length of the documents. In this case, the smallest
document corresponds to smallest symbol in $\col_v$
which is $\Sigma_v[0]$. The latter can be
computed in constant time. Let $v$ be the $i$-th node
of level $\ell$ in $\WTD$ then $\Sigma_v[0]=i\cdot 2^{\lceil\log\docs\rceil-\ell-1}$.
The document lengths are maintained in an 
array $\DOCLEN[0,\docs-1]$.
In our example in Fig. \ref{fig-col} and
\ref{fig-WT} we have reordered the documents using
a permutation $\pi=[1,3,2,0]$. The additional space of
$N\log N+N\log n_{max}$ bits is negligible compared to
the size of the CSA and $\DARRAY$.

\subsection{Improved term frequency estimation}

Unit now we use the range size $\fDvq$ of term $\termt$ in 
$v$ to estimate an upper bound for the maximal term
frequency in a document $d\in \col_v$. 
Knowing the number of distinct documents in $\col_v$,
called $\FDvq$, helps to improve the upper bound
to the number of repetitions plus one: $\deltafF=\fDvq-\FDvq+1$.
In this section, we present a method that computes $\deltafF$ 
in constant time during $\WTDT$ traversal. 

The solution is built on top of Sadakane's \cite{sada07jda} 
document frequency structure (\DFSADA), which solves the problem solely
for $\ccol_v=\ccol$. 
We briefly revisit the structure:
first, a BST is built over
$\ccol$, see Fig.~\ref{fig-df}. The leaves are labeled
with the corresponding
documents, i.e. from left to right \DARRAY\ is formed.
The inner nodes are numbered from $1$ to $n-1$ in-order.
Each node $w_i$ holds a list $\mathcal{R}_{i}$, 
containing all documents which
occur in both subtrees of $w_i$. We refer to elements
in $\mathcal{R}_i$ as \emph{repetitions}.
Let $w_i$ be the locus of a term \termt\ 
in the BST and let $[l,r]$ be $w_i$'s interval.
Then the total number of repetitions in
$\DARRAY[l,r]$ can be calculated by accumulating
the length of all repetition lists in $w_i$'s subtree.
To achieve this, Sadakane generated a 
bitvector $H$ that concatenates the unary coding of 
the lengths of all $\mathcal{R}_i$: 
$H=10^{|\mathcal{R}_{0}|}10^{|\mathcal{R}_{1}|}1\ldots0^{|\mathcal{R}_{n-1}|}1$.
The subtree interval $[l,r]$ can be mapped
into $H$ via select operations: 
$[l',r']=[select_1(l,H),select_1(r,H)]$, since the 
accumulation of the list lengths equals the number of
zeros in $[l',r']$. The following example illustrates
the process:
interval $[4,9]$ corresponds to term $\termt=$\textcolor{brown}{LA}
and is mapped to $[l',r']=[select_1(4,H),select_1(9,H)]=[7,15]$ in $H$. 
It follows that there are $z_l=l'\!-\!l=3$ zeros in $H[0,l']$ and $z_r=r'\!-\!r=6$
in $H[0,r']$; thus there are $6\!-\!3=3$ repetitions in $\DARRAY[4,9]$. 
We can overestimate the maximal term frequency by
assuming that all repetitions belong to the same
document $\docd_x$ and add one for $\docd_x$
itself. So $\deltafF=4$ in this case. This overestimates
the maximal term frequency, which is $f_{d_3,\termt}=3$, 
by one. The interval size estimate would have been $6$

We now extend Sadakane's solution to work on all subsets $\col_v$. 
First, we concatenate all $\mathcal{R}_i$ 
and form the \emph{repetition array} $\REP[0,n\!-\!\docs-1]$ (again, see Fig. \ref{fig-df}),
containing the actual repetition value for each zero in $H$.
As above, using $H$ and $select_1$, we can map $[l,r]$ to
the corresponding range $[l'',r'']=[z_l,z_r-1]$ in $\REP$.
To calculate $\deltafF$ for $\col_v$ we
represent $\REP$ as a WT.
Now, we can traverse $\WTDT$ and $\WTREPT$
simultaneously, mapping $[l'',r'']$ to
$[l'',r'']_{v}$ in $\WTREPT$. The size
of $[l'',r'']_{v}\!+\!1$ equals $\deltafF$ since node $v$
contains only repetitions of $\col_v$.

\section{Space reduction} 
The space of $\REP$ can be reduced to array $\REPHAT$ by omitting all elements
belonging to the root $v_{ST}$ of the non-binary ST\ 
since we will never query the empty string.
In Fig. \ref{fig-df} all nodes with empty
path labels correspond to $v_{ST}$, i.e.
$v_1, v_4,$ and $v_{10}$. 
Hence $\REPHAT=\{3,3,1,2\}$ and we employ a
bitvector to map from the index domain of $\REP$
into $\REPHAT$.

Second, we note that the space of $\WTDT$ and $\WTREPTHAT$ can be
reduced if the length of query phrase is restricted
to length $\ell$. In this case, 
we can sort ranges in $\REPHAT$ which belong to nodes $v_i$,
where $v_i$ are the loci of patterns of length $\ell$.
Since all query ranges are aligned at borders of
sorted ranges, the interval sizes during processing
will not be affected.
In our running example, if $\ell=1$, we can sort
the elements of $v_9$'s subtree, resulting in
$\REPHAT^{1}=\{1,3,3,2\}$. The sorting will result
in better compressibility of $\WTREPHATX$.

Third, we observe that when using $\WTREPHATX$ \emph{only a part
of $\WTDT$} is necessary to calculate $\deltafF$.
If $\termt$ occurs more than once in $\col_v$,
$\WTREPHATX$ can be used to get $\deltafF$.
Hence, $\WTDT$ is only used to determine the
existence of $\termt$ in $\col_v$, and we only store need to store
the unique values inside ranges corresponding
to loci of $\ell$-length patterns. 
In addition, values
in these ranges can be sorted, since this does not change
the result of the existence queries. 
In our example we get $\DARRAY^1=\{3,0,1,2,0,1,2,0,1,2\}$;
one increasing sequence per symbol. A bitvector
is again used to map into $\DARRAY^{\ell}$.

\section{Experimental Study}\label{sec-experiments}

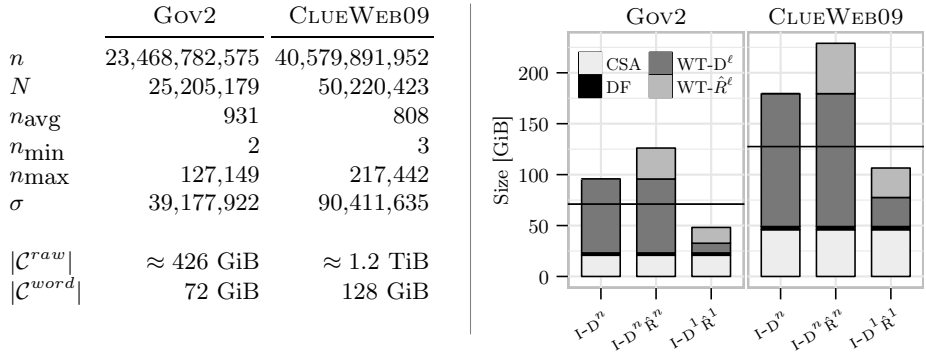
\begin{figure}[tb]
\begin{minipage}[b][4.4cm][t]{0.50\textwidth}
{\begin{tabular}{l@{\hspace{3mm}}rcr}
          & \multicolumn{1}{c}{\govtwo}          && \clueweb \\\cmidrule{2-2}\cmidrule{4-4}
$n$       & $23{,}468{,}782{,}575$ && $40{,}579{,}891{,}952$ \\
$\docs$   & $25{,}205{,}179$     && $50{,}220{,}423$ \\
$n_{\mbox{avg}}$ & $931$            && $808$  \\
$n_{\mbox{min}}$ & $2$              && $3$ \\
$n_{\mbox{max}}$ & $127{,}149$        && $217{,}442$ \\
$\sigma$  & $39{,}177{,}922$     && $90{,}411{,}635$ \\
          & && \\
$|\ccol^{raw}|$  &$\approx 426$ GiB && $\approx 1.2$ TiB\\
$|\ccol^{word}|$ &  $72$ GiB && $128$ GiB\\
\end{tabular}}
\end{minipage}
\begin{minipage}[b][4.4cm][t]{0.01\textwidth}
\begin{tikzpicture}
\draw[gray] (0,0) -- (0,4.4cm);
\end{tikzpicture}
\end{minipage}
\begin{minipage}[b][4.4cm][t]{0.49\textwidth}
\vspace{-0.2em}
\begin{tikzpicture}[x=1pt,y=1pt]
\definecolor[named]{fillColor}{rgb}{1.00,1.00,1.00}
\path[use as bounding box,fill=fillColor,fill opacity=0.00] (0,0) rectangle (166.22,144.54);
\begin{scope}
\path[clip] (  0.00,  0.00) rectangle (166.22,144.54);
\definecolor[named]{fillColor}{rgb}{1.00,1.00,1.00}

\path[fill=fillColor] (  0.00,  0.00) rectangle (166.22,144.54);
\end{scope}
\begin{scope}
\path[clip] ( 31.52,131.91) rectangle ( 98.27,144.54);

\path[] ( 31.52,131.91) rectangle ( 98.27,144.54);
\definecolor[named]{drawColor}{rgb}{0.00,0.00,0.00}

\node[text=drawColor,anchor=base,inner sep=0pt, outer sep=0pt, scale=  0.96] at ( 64.89,134.92) {\govtwo};
\end{scope}
\begin{scope}
\path[clip] ( 99.47,131.91) rectangle (166.22,144.54);

\path[] ( 99.47,131.91) rectangle (166.22,144.54);
\definecolor[named]{drawColor}{rgb}{0.00,0.00,0.00}

\node[text=drawColor,anchor=base,inner sep=0pt, outer sep=0pt, scale=  0.96] at (132.85,134.92) {\clueweb};
\end{scope}
\begin{scope}
\path[clip] ( 31.52, 34.81) rectangle ( 98.27,131.91);
\definecolor[named]{drawColor}{rgb}{0.75,0.75,0.75}

\path[draw=drawColor,line width= 1.5pt,line join=round,line cap=round] ( 31.52, 34.81) rectangle ( 98.27,131.91);
\definecolor[named]{drawColor}{rgb}{0.90,0.90,0.90}

\path[draw=drawColor,line width= 0.3pt,line join=round] ( 31.52, 48.87) --
	( 98.27, 48.87);

\path[draw=drawColor,line width= 0.3pt,line join=round] ( 31.52, 68.14) --
	( 98.27, 68.14);

\path[draw=drawColor,line width= 0.3pt,line join=round] ( 31.52, 87.42) --
	( 98.27, 87.42);

\path[draw=drawColor,line width= 0.3pt,line join=round] ( 31.52,106.70) --
	( 98.27,106.70);

\path[draw=drawColor,line width= 0.3pt,line join=round] ( 31.52,125.97) --
	( 98.27,125.97);

\path[draw=drawColor,line width= 0.8pt,line join=round] ( 31.52, 39.23) --
	( 98.27, 39.23);

\path[draw=drawColor,line width= 0.8pt,line join=round] ( 31.52, 58.50) --
	( 98.27, 58.50);

\path[draw=drawColor,line width= 0.8pt,line join=round] ( 31.52, 77.78) --
	( 98.27, 77.78);

\path[draw=drawColor,line width= 0.8pt,line join=round] ( 31.52, 97.06) --
	( 98.27, 97.06);

\path[draw=drawColor,line width= 0.8pt,line join=round] ( 31.52,116.33) --
	( 98.27,116.33);

\path[draw=drawColor,line width= 0.8pt,line join=round] ( 44.03, 34.81) --
	( 44.03,131.91);

\path[draw=drawColor,line width= 0.8pt,line join=round] ( 64.89, 34.81) --
	( 64.89,131.91);

\path[draw=drawColor,line width= 0.8pt,line join=round] ( 85.75, 34.81) --
	( 85.75,131.91);
\definecolor[named]{fillColor}{rgb}{0.93,0.93,0.93}

\path[fill=fillColor] ( 36.73, 39.23) rectangle ( 51.33, 47.25);
\definecolor[named]{fillColor}{rgb}{0.00,0.00,0.00}

\path[fill=fillColor] ( 36.73, 47.25) rectangle ( 51.33, 48.14);
\definecolor[named]{fillColor}{rgb}{0.47,0.47,0.47}

\path[fill=fillColor] ( 36.73, 48.14) rectangle ( 51.33, 76.12);
\definecolor[named]{fillColor}{rgb}{0.73,0.73,0.73}

\path[fill=fillColor] ( 36.73, 76.12) rectangle ( 51.33, 76.12);
\definecolor[named]{fillColor}{rgb}{0.93,0.93,0.93}

\path[fill=fillColor] ( 57.59, 39.23) rectangle ( 72.19, 47.25);
\definecolor[named]{fillColor}{rgb}{0.00,0.00,0.00}

\path[fill=fillColor] ( 57.59, 47.25) rectangle ( 72.19, 48.14);
\definecolor[named]{fillColor}{rgb}{0.47,0.47,0.47}

\path[fill=fillColor] ( 57.59, 48.14) rectangle ( 72.19, 76.12);
\definecolor[named]{fillColor}{rgb}{0.73,0.73,0.73}

\path[fill=fillColor] ( 57.59, 76.12) rectangle ( 72.19, 87.84);
\definecolor[named]{fillColor}{rgb}{0.93,0.93,0.93}

\path[fill=fillColor] ( 78.45, 39.23) rectangle ( 93.05, 47.25);
\definecolor[named]{fillColor}{rgb}{0.00,0.00,0.00}

\path[fill=fillColor] ( 78.45, 47.25) rectangle ( 93.05, 48.14);
\definecolor[named]{fillColor}{rgb}{0.47,0.47,0.47}

\path[fill=fillColor] ( 78.45, 48.14) rectangle ( 93.05, 51.81);
\definecolor[named]{fillColor}{rgb}{0.73,0.73,0.73}

\path[fill=fillColor] ( 78.45, 51.81) rectangle ( 93.05, 57.79);
\definecolor[named]{drawColor}{rgb}{0.00,0.00,0.00}
\definecolor[named]{fillColor}{rgb}{0.93,0.93,0.93}

\path[draw=drawColor,line width= 0.6pt,line join=round,fill=fillColor] ( 36.73, 39.23) rectangle ( 51.33, 47.25);
\definecolor[named]{fillColor}{rgb}{0.00,0.00,0.00}

\path[draw=drawColor,line width= 0.6pt,line join=round,fill=fillColor] ( 36.73, 47.25) rectangle ( 51.33, 48.14);
\definecolor[named]{fillColor}{rgb}{0.47,0.47,0.47}

\path[draw=drawColor,line width= 0.6pt,line join=round,fill=fillColor] ( 36.73, 48.14) rectangle ( 51.33, 76.12);
\definecolor[named]{fillColor}{rgb}{0.73,0.73,0.73}

\path[draw=drawColor,line width= 0.6pt,line join=round,fill=fillColor] ( 36.73, 76.12) rectangle ( 51.33, 76.12);
\definecolor[named]{fillColor}{rgb}{0.93,0.93,0.93}

\path[draw=drawColor,line width= 0.6pt,line join=round,fill=fillColor] ( 57.59, 39.23) rectangle ( 72.19, 47.25);
\definecolor[named]{fillColor}{rgb}{0.00,0.00,0.00}

\path[draw=drawColor,line width= 0.6pt,line join=round,fill=fillColor] ( 57.59, 47.25) rectangle ( 72.19, 48.14);
\definecolor[named]{fillColor}{rgb}{0.47,0.47,0.47}

\path[draw=drawColor,line width= 0.6pt,line join=round,fill=fillColor] ( 57.59, 48.14) rectangle ( 72.19, 76.12);
\definecolor[named]{fillColor}{rgb}{0.73,0.73,0.73}

\path[draw=drawColor,line width= 0.6pt,line join=round,fill=fillColor] ( 57.59, 76.12) rectangle ( 72.19, 87.84);
\definecolor[named]{fillColor}{rgb}{0.93,0.93,0.93}

\path[draw=drawColor,line width= 0.6pt,line join=round,fill=fillColor] ( 78.45, 39.23) rectangle ( 93.05, 47.25);
\definecolor[named]{fillColor}{rgb}{0.00,0.00,0.00}

\path[draw=drawColor,line width= 0.6pt,line join=round,fill=fillColor] ( 78.45, 47.25) rectangle ( 93.05, 48.14);
\definecolor[named]{fillColor}{rgb}{0.47,0.47,0.47}

\path[draw=drawColor,line width= 0.6pt,line join=round,fill=fillColor] ( 78.45, 48.14) rectangle ( 93.05, 51.81);
\definecolor[named]{fillColor}{rgb}{0.73,0.73,0.73}

\path[draw=drawColor,line width= 0.6pt,line join=round,fill=fillColor] ( 78.45, 51.81) rectangle ( 93.05, 57.79);
\definecolor[named]{fillColor}{rgb}{0.00,0.00,0.00}

\path[draw=drawColor,line width= 0.6pt,line join=round,fill=fillColor] ( 31.52, 66.61) -- ( 98.27, 66.61);
\end{scope}
\begin{scope}
\path[clip] ( 99.47, 34.81) rectangle (166.22,131.91);
\definecolor[named]{drawColor}{rgb}{0.75,0.75,0.75}

\path[draw=drawColor,line width= 1.5pt,line join=round,line cap=round] ( 99.47, 34.81) rectangle (166.22,131.91);
\definecolor[named]{drawColor}{rgb}{0.90,0.90,0.90}

\path[draw=drawColor,line width= 0.3pt,line join=round] ( 99.47, 48.87) --
	(166.22, 48.87);

\path[draw=drawColor,line width= 0.3pt,line join=round] ( 99.47, 68.14) --
	(166.22, 68.14);

\path[draw=drawColor,line width= 0.3pt,line join=round] ( 99.47, 87.42) --
	(166.22, 87.42);

\path[draw=drawColor,line width= 0.3pt,line join=round] ( 99.47,106.70) --
	(166.22,106.70);

\path[draw=drawColor,line width= 0.3pt,line join=round] ( 99.47,125.97) --
	(166.22,125.97);

\path[draw=drawColor,line width= 0.8pt,line join=round] ( 99.47, 39.23) --
	(166.22, 39.23);

\path[draw=drawColor,line width= 0.8pt,line join=round] ( 99.47, 58.50) --
	(166.22, 58.50);

\path[draw=drawColor,line width= 0.8pt,line join=round] ( 99.47, 77.78) --
	(166.22, 77.78);

\path[draw=drawColor,line width= 0.8pt,line join=round] ( 99.47, 97.06) --
	(166.22, 97.06);

\path[draw=drawColor,line width= 0.8pt,line join=round] ( 99.47,116.33) --
	(166.22,116.33);

\path[draw=drawColor,line width= 0.8pt,line join=round] (111.99, 34.81) --
	(111.99,131.91);

\path[draw=drawColor,line width= 0.8pt,line join=round] (132.85, 34.81) --
	(132.85,131.91);

\path[draw=drawColor,line width= 0.8pt,line join=round] (153.71, 34.81) --
	(153.71,131.91);
\definecolor[named]{fillColor}{rgb}{0.93,0.93,0.93}

\path[fill=fillColor] (104.69, 39.23) rectangle (119.29, 56.93);
\definecolor[named]{fillColor}{rgb}{0.00,0.00,0.00}

\path[fill=fillColor] (104.69, 56.93) rectangle (119.29, 58.12);
\definecolor[named]{fillColor}{rgb}{0.47,0.47,0.47}

\path[fill=fillColor] (104.69, 58.12) rectangle (119.29,108.43);
\definecolor[named]{fillColor}{rgb}{0.73,0.73,0.73}

\path[fill=fillColor] (104.69,108.43) rectangle (119.29,108.43);
\definecolor[named]{fillColor}{rgb}{0.93,0.93,0.93}

\path[fill=fillColor] (125.55, 39.23) rectangle (140.15, 56.93);
\definecolor[named]{fillColor}{rgb}{0.00,0.00,0.00}

\path[fill=fillColor] (125.55, 56.93) rectangle (140.15, 58.12);
\definecolor[named]{fillColor}{rgb}{0.47,0.47,0.47}

\path[fill=fillColor] (125.55, 58.12) rectangle (140.15,108.43);
\definecolor[named]{fillColor}{rgb}{0.73,0.73,0.73}

\path[fill=fillColor] (125.55,108.43) rectangle (140.15,127.49);
\definecolor[named]{fillColor}{rgb}{0.93,0.93,0.93}

\path[fill=fillColor] (146.40, 39.23) rectangle (161.01, 56.93);
\definecolor[named]{fillColor}{rgb}{0.00,0.00,0.00}

\path[fill=fillColor] (146.40, 56.93) rectangle (161.01, 58.12);
\definecolor[named]{fillColor}{rgb}{0.47,0.47,0.47}

\path[fill=fillColor] (146.40, 58.12) rectangle (161.01, 69.12);
\definecolor[named]{fillColor}{rgb}{0.73,0.73,0.73}

\path[fill=fillColor] (146.40, 69.12) rectangle (161.01, 80.28);
\definecolor[named]{drawColor}{rgb}{0.00,0.00,0.00}
\definecolor[named]{fillColor}{rgb}{0.93,0.93,0.93}

\path[draw=drawColor,line width= 0.6pt,line join=round,fill=fillColor] (104.69, 39.23) rectangle (119.29, 56.93);
\definecolor[named]{fillColor}{rgb}{0.00,0.00,0.00}

\path[draw=drawColor,line width= 0.6pt,line join=round,fill=fillColor] (104.69, 56.93) rectangle (119.29, 58.12);
\definecolor[named]{fillColor}{rgb}{0.47,0.47,0.47}

\path[draw=drawColor,line width= 0.6pt,line join=round,fill=fillColor] (104.69, 58.12) rectangle (119.29,108.43);
\definecolor[named]{fillColor}{rgb}{0.73,0.73,0.73}

\path[draw=drawColor,line width= 0.6pt,line join=round,fill=fillColor] (104.69,108.43) rectangle (119.29,108.43);
\definecolor[named]{fillColor}{rgb}{0.93,0.93,0.93}

\path[draw=drawColor,line width= 0.6pt,line join=round,fill=fillColor] (125.55, 39.23) rectangle (140.15, 56.93);
\definecolor[named]{fillColor}{rgb}{0.00,0.00,0.00}

\path[draw=drawColor,line width= 0.6pt,line join=round,fill=fillColor] (125.55, 56.93) rectangle (140.15, 58.12);
\definecolor[named]{fillColor}{rgb}{0.47,0.47,0.47}

\path[draw=drawColor,line width= 0.6pt,line join=round,fill=fillColor] (125.55, 58.12) rectangle (140.15,108.43);
\definecolor[named]{fillColor}{rgb}{0.73,0.73,0.73}

\path[draw=drawColor,line width= 0.6pt,line join=round,fill=fillColor] (125.55,108.43) rectangle (140.15,127.49);
\definecolor[named]{fillColor}{rgb}{0.93,0.93,0.93}

\path[draw=drawColor,line width= 0.6pt,line join=round,fill=fillColor] (146.40, 39.23) rectangle (161.01, 56.93);
\definecolor[named]{fillColor}{rgb}{0.00,0.00,0.00}

\path[draw=drawColor,line width= 0.6pt,line join=round,fill=fillColor] (146.40, 56.93) rectangle (161.01, 58.12);
\definecolor[named]{fillColor}{rgb}{0.47,0.47,0.47}

\path[draw=drawColor,line width= 0.6pt,line join=round,fill=fillColor] (146.40, 58.12) rectangle (161.01, 69.12);
\definecolor[named]{fillColor}{rgb}{0.73,0.73,0.73}

\path[draw=drawColor,line width= 0.6pt,line join=round,fill=fillColor] (146.40, 69.12) rectangle (161.01, 80.28);
\definecolor[named]{fillColor}{rgb}{0.00,0.00,0.00}

\path[draw=drawColor,line width= 0.6pt,line join=round,fill=fillColor] ( 99.47, 88.40) -- (166.22, 88.40);
\end{scope}
\begin{scope}
\path[clip] (  0.00,  0.00) rectangle (166.22,144.54);
\definecolor[named]{drawColor}{rgb}{0.00,0.00,0.00}

\node[text=drawColor,anchor=base east,inner sep=0pt, outer sep=0pt, scale=  0.84] at ( 24.40, 36.34) {0};

\node[text=drawColor,anchor=base east,inner sep=0pt, outer sep=0pt, scale=  0.84] at ( 24.40, 55.61) {50};

\node[text=drawColor,anchor=base east,inner sep=0pt, outer sep=0pt, scale=  0.84] at ( 24.40, 74.89) {100};

\node[text=drawColor,anchor=base east,inner sep=0pt, outer sep=0pt, scale=  0.84] at ( 24.40, 94.17) {150};

\node[text=drawColor,anchor=base east,inner sep=0pt, outer sep=0pt, scale=  0.84] at ( 24.40,113.44) {200};
\end{scope}
\begin{scope}
\path[clip] (  0.00,  0.00) rectangle (166.22,144.54);
\definecolor[named]{drawColor}{rgb}{0.00,0.00,0.00}

\path[draw=drawColor,line width= 0.6pt,line join=round] ( 27.25, 39.23) --
	( 31.52, 39.23);

\path[draw=drawColor,line width= 0.6pt,line join=round] ( 27.25, 58.50) --
	( 31.52, 58.50);

\path[draw=drawColor,line width= 0.6pt,line join=round] ( 27.25, 77.78) --
	( 31.52, 77.78);

\path[draw=drawColor,line width= 0.6pt,line join=round] ( 27.25, 97.06) --
	( 31.52, 97.06);

\path[draw=drawColor,line width= 0.6pt,line join=round] ( 27.25,116.33) --
	( 31.52,116.33);
\end{scope}
\begin{scope}
\path[clip] (  0.00,  0.00) rectangle (166.22,144.54);
\definecolor[named]{drawColor}{rgb}{0.00,0.00,0.00}

\path[draw=drawColor,line width= 0.6pt,line join=round] ( 44.03, 30.55) --
	( 44.03, 34.81);

\path[draw=drawColor,line width= 0.6pt,line join=round] ( 64.89, 30.55) --
	( 64.89, 34.81);

\path[draw=drawColor,line width= 0.6pt,line join=round] ( 85.75, 30.55) --
	( 85.75, 34.81);
\end{scope}
\begin{scope}
\path[clip] (  0.00,  0.00) rectangle (166.22,144.54);
\definecolor[named]{drawColor}{rgb}{0.00,0.00,0.00}

\node[text=drawColor,rotate= 30.00,anchor=base east,inner sep=0pt, outer sep=0pt, scale=  0.84] at ( 47.98, 22.21) {\idxd};

\node[text=drawColor,rotate= 30.00,anchor=base east,inner sep=0pt, outer sep=0pt, scale=  0.84] at ( 69.67, 22.69) {\idxdr};

\node[text=drawColor,rotate= 30.00,anchor=base east,inner sep=0pt, outer sep=0pt, scale=  0.84] at ( 90.39, 22.61) {\idxdIrI};
\end{scope}
\begin{scope}
\path[clip] (  0.00,  0.00) rectangle (166.22,144.54);
\definecolor[named]{drawColor}{rgb}{0.00,0.00,0.00}

\path[draw=drawColor,line width= 0.6pt,line join=round] (111.99, 30.55) --
	(111.99, 34.81);

\path[draw=drawColor,line width= 0.6pt,line join=round] (132.85, 30.55) --
	(132.85, 34.81);

\path[draw=drawColor,line width= 0.6pt,line join=round] (153.71, 30.55) --
	(153.71, 34.81);
\end{scope}
\begin{scope}
\path[clip] (  0.00,  0.00) rectangle (166.22,144.54);
\definecolor[named]{drawColor}{rgb}{0.00,0.00,0.00}

\node[text=drawColor,rotate= 30.00,anchor=base east,inner sep=0pt, outer sep=0pt, scale=  0.84] at (115.94, 22.21) {\idxd};

\node[text=drawColor,rotate= 30.00,anchor=base east,inner sep=0pt, outer sep=0pt, scale=  0.84] at (137.63, 22.69) {\idxdr};

\node[text=drawColor,rotate= 30.00,anchor=base east,inner sep=0pt, outer sep=0pt, scale=  0.84] at (158.35, 22.61) {\idxdIrI};
\end{scope}
\begin{scope}
\path[clip] (  0.00,  0.00) rectangle (166.22,144.54);
\definecolor[named]{drawColor}{rgb}{0.00,0.00,0.00}

\node[text=drawColor,rotate= 90.00,anchor=base,inner sep=0pt, outer sep=0pt, scale=  0.84] at (  8.80, 83.36) {Size [GiB]};
\end{scope}
\begin{scope}
\path[clip] (  0.00,  0.00) rectangle (166.22,144.54);
\definecolor[named]{drawColor}{rgb}{0.00,0.00,0.00}
\definecolor[named]{fillColor}{rgb}{1.00,1.00,1.00}

\path[draw=drawColor,line width= 0.6pt,line join=round,line cap=round,fill=fillColor] ( 35.67,115.54) rectangle ( 44.20,124.07);
\end{scope}
\begin{scope}
\path[clip] (  0.00,  0.00) rectangle (166.22,144.54);
\definecolor[named]{fillColor}{rgb}{0.93,0.93,0.93}

\path[fill=fillColor] ( 35.67,115.54) rectangle ( 44.20,124.07);

\path[] ( 35.67,115.54) --
	( 44.20,124.07);
\end{scope}
\begin{scope}
\path[clip] (  0.00,  0.00) rectangle (166.22,144.54);
\definecolor[named]{drawColor}{rgb}{0.00,0.00,0.00}
\definecolor[named]{fillColor}{rgb}{1.00,1.00,1.00}

\path[draw=drawColor,line width= 0.6pt,line join=round,line cap=round,fill=fillColor] ( 35.67,107.00) rectangle ( 44.20,115.54);
\end{scope}
\begin{scope}
\path[clip] (  0.00,  0.00) rectangle (166.22,144.54);
\definecolor[named]{fillColor}{rgb}{0.00,0.00,0.00}

\path[fill=fillColor] ( 35.67,107.00) rectangle ( 44.20,115.54);

\path[] ( 35.67,107.00) --
	( 44.20,115.54);
\end{scope}
\begin{scope}
\path[clip] (  0.00,  0.00) rectangle (166.22,144.54);
\definecolor[named]{drawColor}{rgb}{0.00,0.00,0.00}
\definecolor[named]{fillColor}{rgb}{1.00,1.00,1.00}

\path[draw=drawColor,line width= 0.6pt,line join=round,line cap=round,fill=fillColor] ( 62.41,115.54) rectangle ( 70.95,124.07);
\end{scope}
\begin{scope}
\path[clip] (  0.00,  0.00) rectangle (166.22,144.54);
\definecolor[named]{fillColor}{rgb}{0.47,0.47,0.47}

\path[fill=fillColor] ( 62.41,115.54) rectangle ( 70.95,124.07);

\path[] ( 62.41,115.54) --
	( 70.95,124.07);
\end{scope}
\begin{scope}
\path[clip] (  0.00,  0.00) rectangle (166.22,144.54);
\definecolor[named]{drawColor}{rgb}{0.00,0.00,0.00}
\definecolor[named]{fillColor}{rgb}{1.00,1.00,1.00}

\path[draw=drawColor,line width= 0.6pt,line join=round,line cap=round,fill=fillColor] ( 62.41,107.00) rectangle ( 70.95,115.54);
\end{scope}
\begin{scope}
\path[clip] (  0.00,  0.00) rectangle (166.22,144.54);
\definecolor[named]{fillColor}{rgb}{0.73,0.73,0.73}

\path[fill=fillColor] ( 62.41,107.00) rectangle ( 70.95,115.54);

\path[] ( 62.41,107.00) --
	( 70.95,115.54);
\end{scope}
\begin{scope}
\path[clip] (  0.00,  0.00) rectangle (166.22,144.54);
\definecolor[named]{drawColor}{rgb}{0.00,0.00,0.00}

\node[text=drawColor,anchor=base west,inner sep=0pt, outer sep=0pt, scale=  0.72] at ( 46.01,117.32) {\CSA};
\end{scope}
\begin{scope}
\path[clip] (  0.00,  0.00) rectangle (166.22,144.54);
\definecolor[named]{drawColor}{rgb}{0.00,0.00,0.00}

\node[text=drawColor,anchor=base west,inner sep=0pt, outer sep=0pt, scale=  0.72] at ( 46.01,108.79) {\DFSADA};
\end{scope}
\begin{scope}
\path[clip] (  0.00,  0.00) rectangle (166.22,144.54);
\definecolor[named]{drawColor}{rgb}{0.00,0.00,0.00}

\node[text=drawColor,anchor=base west,inner sep=0pt, outer sep=0pt, scale=  0.72] at ( 72.75,117.32) {\WTDX};
\end{scope}
\begin{scope}
\path[clip] (  0.00,  0.00) rectangle (166.22,144.54);
\definecolor[named]{drawColor}{rgb}{0.00,0.00,0.00}

\node[text=drawColor,anchor=base west,inner sep=0pt, outer sep=0pt, scale=  0.72] at ( 72.75,108.79) {\WTREPHATX};
\end{scope}
\end{tikzpicture}
\end{minipage}
\caption{Collection statistics for {\govtwo} and {\clueweb} (left) and 
memory breakdown of our indexes (right). $|\ccol^{raw}|$ denotes the original size of
the collection, while $|\ccol^{word}|$ denotes the size after parsing it.
A more detailed space breakdown of the indexes is available at
\href{http://go.unimelb.edu.au/6a4n}{http://go.unimelb.edu.au/6a4n}.}
\label{fig:statistics}
\end{figure}

\emph{Indexes and Implementations.} To evaluate our proposals we created the
SUccinct Retrieval Framework (\surf)\footnote{Source code, test
queries, and the scripts to reproduce the experiments are publicly available at
\url{https://github.com/simongog/surf/}.%
} 
which implements document retrieval specific components, like
Sadakane's \DFSADA\ structure. These components can be parametrized
by structures provided by the \sdsl\ library~\cite{gbmp2014sea}.
We assembled three self-index based systems,
corresponding to different functionality-space trade-offs.
All systems use the same \CSA\ and \DFSADA\ structure. The \CSA\
is implemented as an FM-index using a WT. This WT
as well as \DFSADA\ use a practical compressed bitvector~\cite{NAV:PRO:2012}
to minimize space, since all query related operations
on these components take only a fraction of a millisecond. 

Our first index (\idxd) adds \WTD, which uses an uncompressed bitvector
to allow fast \WT\ traversal. Our second structure (\idxdr) 
adds \WTREPHAT. It uses a compressed bitvector to compress the
increasing sequences in $\REPHATN$. To show a functionality-space trade-off
we also provide the previous index with a phrase length restriction of one,
named \idxdIrI.
The components $\WTDI$ and $\WTREPHATI$ both use compressed bitvectors
to minimize the space of the WTs.

As a reference point we also implemented a
competitive inverted file based search index (\invidx) which stores block-based
postings lists compressed using \method{OptPFD}~\cite{lb13spe,yds09www}. For 
each block, a representative is stored to allow efficient skipping. The
document ranking is calculated using two
list processing schemes. The first scheme -- \invidxwand\ -- uses the efficient
\method{Wand} list processing algorithm~\cite{bch+03-cikm}. However,
\method{Wand} and other advanced processing schemes require similarity measure 
specific pre-computation during construction time. A more flexible, but
less efficient processing scheme -- \invidxexh\ -- exhaustively evaluates
all postings in document-at-a-time order without either the burden or benefit of
score pre-computation.

\emph{Data Sets, Queries and Test Environment.}
We compare our index structures over two standard IR test collections: {\govtwo} 
and {\clueweb}. {\govtwo} is the test collection 
of the TREC 2004 Terabyte Track competition
and the {\clueweb} collection consists of the English text ``Category B'' subset 
of the ClueWeb09 dataset\footnote{http://lemurproject.org/clueweb09/}. 
To ensure reproducibility we extract the integer token sequence \colC\ for both 
collections from Indri\footnote{http://www.lemurproject.org/indri/} using default
parameters. 
No stop-words were removed from the collection. We evaluate our indexes 
using queries chosen from the TREC 2005 and TREC 2006 Terabyte 
track efficiency queries\footnote{http://trec.nist.gov/data/terabyte.html}. A total of
$1000$ queries were randomly sampled from both query sets, ensuring all query
terms are present in both test collections. Statistics for both collections,
given in Fig~\ref{fig:statistics} (left), are in line with other studies~\cite{v-wsdm13}.
We support ranked disjunctive (\rankor, at least one term must occur) 
and ranked conjunctive (\rankand, all terms must occur) retrieval. 
All indexes were implemented using C++11 and
compiled using GCC 4.8.1 with optimizations. Our machine was equipped with \gb{256} RAM and
we used one Intel Sandybridge core (E5-2680) running at $2.7$ Ghz. All indexes were
loaded into main memory prior to query processing.

\emph{Space Usage.} The space usage of our indexes is 
summarized in Fig.~\ref{fig:statistics} (right).
Note that our smallest index is 4 to 5 times larger than 
our reference inverted index. 
However, an
inverted index supporting phrase queries would require additional positional 
information, which would significantly increase its size.
The size of our integer parsing of size $n\lceil\log\sigma\rceil$ is shown as a horizontal line. 
The \CSA\ for both collections compresses
to roughly $30\%$ of the size of the integer parsing. The space for \DFSADA\ 
is negligible. The \WTD\ has the size of the integer parsing plus
$5\%$ overhead for a rank structure.
The size reduction from \REP\ to \REPHAT\ is substantial. For example, 
storing \REP\ for \clueweb\ requires \gb{123}, whereas \REPHAT\ requires 
only \gb{74}. Restricting the phrase length to one (\idxdIrI),
which makes it equivalent to a non-positional inverted index,
shrinks the structure below the original input size.

\begin{figure}[tb]
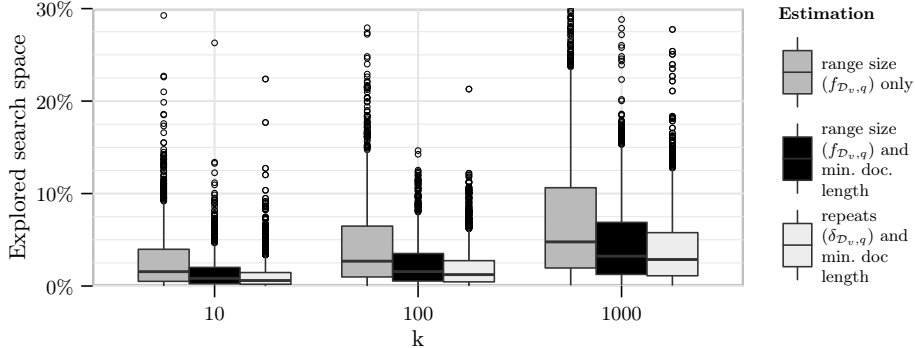

\centering

\vspace{-1cm}
\caption{Percentage of states evaluated $k=10, 100,$ and $1000$ during \rankor retrieval
using \bm for both \trecqryA\ and \trecqryB\ query sets for \govtwo.}
\label{fig:nodes_evaled}
\end{figure}

\emph{Processed States.}
In the first experiment, we measure the quantitative effects of
our improved score estimation during \GREEDY\ processing. 
We compare the range size ($\fDvq$)-only estimation to (a) range size 
estimation including document length estimation and (b) repeats
estimation ($\deltafF$) including document length estimation. 
Fig.~\ref{fig:nodes_evaled} shows the percentage of processed
states for all methods and $k=\{10, 100, 1000\}$ for both query sets 
on \govtwo\ using \bm\ \rankor\ processing. 
The percentage is calculated as the fraction of
states processed compared to the exhaustive processing of each query ($k=\docs$).
For all $k$, range size only estimation evaluates the most states on average.
For $k=10$, the median percentage of evaluated states for range size only
estimation is $1.6\%$. Adding document length estimation reduces the number of
evaluated states by half to $0.8\%$. Using $\deltafF$ instead of $\fDvq$ to estimate
the frequency further improved the percentage of evaluated states to $0.06\%$.
Similar effects can be observed for $k=100$ and $k=1000$. For $k=1000$, document
length estimation reduces the percentage from $5.1\%$ to $3.2\%$. Frequency
estimation using $\deltafF$ again marginally improves the number of evaluated
nodes to $2.8\%$. Overall, document length estimation has a larger impact on
\GREEDY\ than better frequency estimation via $\deltafF$. 

\emph{Disjunctive Ranked Retrieval.}
Next we evaluate the runtime performance of our indexes
\idxd, \idxdr, \idxdIrI\ for 
\bm\ \rankor query processing. Fig.~\ref{fig:run_times_or} (left)
shows runtime on \govtwo\ and both query sets for $k=\{10, 100, 1000\}$. 
We additionally included \invidxwand\ as a reference point for an efficient
inverted index. The latter uses additional
similarity measure dependent information and clearly outperforms all self-index
based indexes.
For $k=10$, it achieves a median runtime performance of less
than $20$ milliseconds, and performs well for other test cases.
Our fastest index, \idxd, is roughly $15$ times slower, achieving a median
runtime of $300$ milliseconds for $k=10$. The other two
indexes, \idxdr\ and \idxdIrI\ are approximately two times slower than \idxd.
This can be explained by the fact that \idxd\ uses an uncompressed \WT,
whereas the other indexes use compressed WTs to save space. Also note that
\idxdIrI\ is faster than \idxdr\ as ranges in \REPHAT$^1$ can be sorted,
which creates runs in the \WT\ which in turn allows faster state processing.
The mean time per processed state -- depicted in Fig.~\ref{fig:run_times_or} (right) --
highlights this observation.
For \idxd, the time linearly increases from $2$ to $5$ microseconds.
While there is a correlation to the number of query terms, rank operations occur
in close proximity -- cache friendly -- within \WTD, which increases performance. 
For the other indexes, we simultaneously access two WTs at different locations
to evaluate a single state. This reduces caching effects resulting in
a stronger dependence on query length.

\begin{figure}[tb]
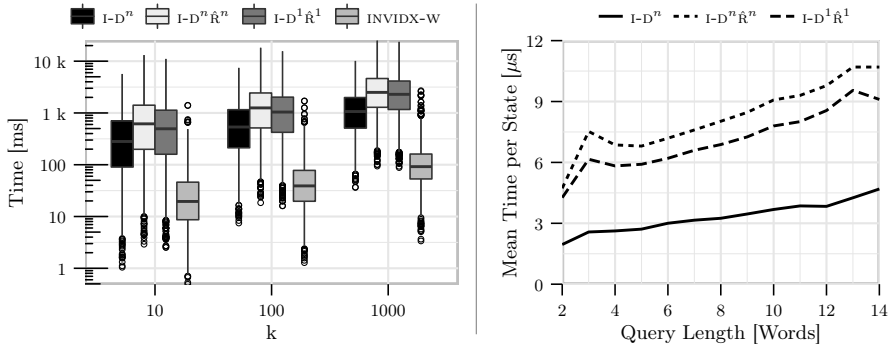

\begin{minipage}[b][4.4cm][t]{0.51\textwidth}

\end{minipage}
\caption{Runtime in milliseconds (left) and time to process one \WT\ state 
for $k=100$ (right) for \bm\ \rankor\ retrieval on \govtwo.}
\label{fig:run_times_or}
\end{figure}

\emph{Efficient Retrieval using Multi-Word Expressions.}
Next we demonstrate one example of additional functionality provided by
our self-indexed based systems. We use the concept of {\it strong associativity}~\cite{ch-cl90} which
defines the ratio of joint probability against the probabilities of a random co-occurrence
as an indicator of a multi word expression (MWE). We create MWEs on-the-fly during query time using the
text statistics provided by the CSA. We use a simple scheme
which greedily parses each query into MWEs. For
example, instead of processing the terms ``saudi'' and ``arabia''
independently, we instead process the MWE ``saudi arabia''
as a single query term. This can significantly reduce the query time of our
indexing schemes. We explore the efficiency of such a ``phrase'' processing
scheme in Fig.~\ref{fig:additional_functionality} (left).
The figure shows the runtime for queries from the \trecqryB\ query sets 
for \govtwo\ using \idxd. For all $k$, the runtime is reduced by an 
order of magnitude. This experiment shows how our system would support
retrieval tasks where the vocabulary does not consist of words but
a large number of entities, since MWE capture the latter.
Supporting MWE does not increase the size of our index,
but would substantially increase the size of an inverted index.

\begin{figure}[tb]
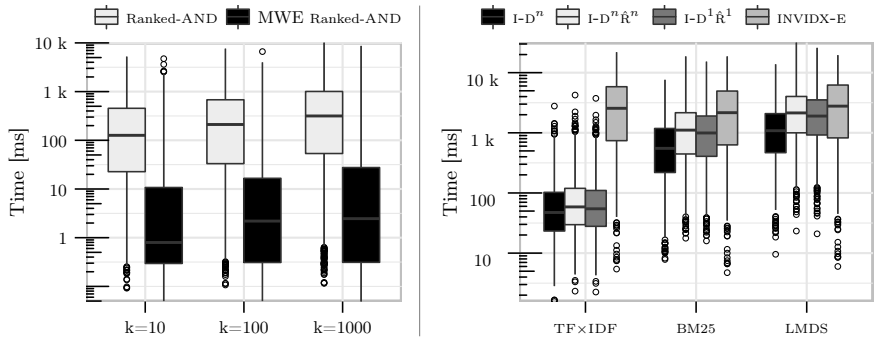

\begin{minipage}[b][4.4cm][t]{0.45\textwidth}

\end{minipage}
\caption{%
\rankand~\bm runtime for unparsed and MWE-parsed queries (right)
\rankor\ runtime for different similarity measures and indexes. 
}
\label{fig:additional_functionality}
\end{figure}

\emph{Flexible Ranked Document Retrieval.}
Another virtue of our proposal is scoring flexibility.
The indexes efficiently support a wide range of 
similarity measures, which can be changed and tuned
after the index is built, while optimized inverted indexes require
similarity measure-dependent pre-computation during
construction~\cite{bch+03-cikm}. 
If ranking functions are only chosen at query time,
inverted indexes exhaustively process postings lists to retrieve the top-$k$
documents. To highlight the benefit of flexibility, we compare our index
structures to \invidxexh\ using three different ranking formulas: \tfxidf, \bm,
and \lmds. Fig.~\ref{fig:additional_functionality} (right) shows the results
on \govtwo.
Interestingly, our index structures significantly outperform the exhaustive 
inverted index for \tfxidf. This can be attributed to the large influence
of the document length $n_\docd$ on $\score^{\tfxidf}$. Unlike \bm\
or \lmds, the final document score is linearly proportional to the actual
size of the document, thus document length estimation significantly reduces
the number of evaluated states. For \bm, the document length contribution
to the final document score is scaled in reference to the average document
length,
and thus has a smaller effect on the overall score of each document. While
our indexes still outperform \invidxexh, the difference is less significant
than for \tfxidf. 

\section{Conclusions and Future Work} \label{sec-conclusion}
We have presented a versatile self-index based retrieval framework
which allows rank-safe top-$k$ retrieval on multi-term queries
using complex scoring functions.
The proposed estimation methods have substantially
improved the query speed compared
to frequency-only score estimation. 
In our experiments we found that top-$k$ document retrieval
is still solved more efficiently by
inverted indexes, if augmented by similarity
measure-dependent pre-computations. 
However, self-index based systems provide
more functionality and thus can be used
in scenarios where the inverted index is not
applicable or slower.
We believe that \GREEDY\ performance can 
be further improved, e.g. by incorporating
range majority queries into the score
estimation.
We provide our framework as open-source
to enable researchers of different research
disciplines to profit from the functionality
provided by self-index based search systems.

\section*{Acknowledgments} 
We are grateful to Paul Cook, who pointed us to \cite{ch-cl90},
and Alistair Moffat and Andrew Turpin for fixing our grammar.
This research was supported by a Victorian Life Sciences
Computation Initiative~(VLSCI) grant number VR0052 on its Peak
Computing Facility at the University of Melbourne, an
initiative of the Victorian Government.
Both authors were funded by ARC DP grant DP110101743.

\bibliographystyle{splncs}
\bibliography{local}

\newpage
\appendix

\section{Correctness of \GREEDY}

\begin{lemma}
Given a document collection $\col$, a query $\qryq$,
and a similarity measure $\score$. Algorithm \GREEDY\ reports
the top-$k$ documents for $\qryq$ of $\col$ if
\begin{itemize}
    \item at every node $v$ the score estimate $s_v$ is 
          not smaller than the maximum document score
          in its subtree
    \item and the score estimates $s_u$ and $s_w$
          of $v$'s children is never larger than
          $s_v$.

\end{itemize}    
\end{lemma}
\begin{proof}
Assume that there is a unreported document $\docd'$ which has
a score larger than the $k$-th reported document $\docd$; i.e.
$s_{v'}>s_{v}$.
Then $\docd'$ was not in the queue when $\docd$ was
reported, since otherwise $\docd'$ would have been reported
first. Therefore an ancestor $\docd''$ of $\docd'$ has to be in the queue.
This is not possible since when $\docd$ was extracted all
score estimates in the queue were smaller or equal to the
score estimate $s_v$ of $\docd$. But the score estimate $s_{v''}$ of
$\docd''$ is larger or equal then $s_{v'}$ and hence larger
than $s_{v}$. This is a contradiction.
\qed
\end{proof}

\section{Additional Similarity Measures}

A simple $\tfxidf$ formulation given in the survey paper of Zobel and Moffat \cite{zm06compsurv}:

\begin{equation}
    S_{\qryq,\docd}^{\tfxidf}=%
        \frac{1}{n_\docd} %
        \sum_{\termt\in \qryq} \underbrace{\left(1+\ln\fdq\right)}_{=\wdt} \cdot %
        \underbrace{\ln\left(1+\frac{\docs}{\FDq}\right)}_{=\wqt}
\label{eq-tfxidf}        
\end{equation}

Another similarity function is based on an Language Model ($\lmds$) formulation:
\begin{equation}
    S_{\qryq,\docd}^{\LM}=%
        m\cdot\ln\left(\frac{\mu}{n_{\docd}+\mu}\right) +
        \sum_{\termt\in \qryq}%
            \underbrace{
            \ln\left(\frac{\fdq}{\mu}\cdot\frac{\tokens}{\FDq}+1\right)    %
            }_{=\wdt} \cdot %
        \underbrace{\fqt}_{=\wqt}
\end{equation}

Parameter $\mu$ is typically set to $2,\!500$.

\section{Collection Statistics and Experimental Results.}

Examples of multi-word expressions (MWEs) detected using
strong associativity and text statistics provided by the \CSA. 
Each detected MWE is shown in brackets.

\begin{multicols}{2}
\begin{itemize}
    \item (fort myers florida) (blue crown conure) (bird breeder) 
    \item map of (saudi arabia)
    \item the (sisterhood of the travel pants) movie
    \item (h1 b visa)
    \item (first tennessee bank) (web site)
    \item (are hot dogs) (healthy food) 
    \item firex (smoke alarm) (downer grove illinois) 
    \item (bayside raider) queens (football coach) 
\end{itemize}
\end{multicols}

Document statistics for both test collections are shown in Figure~\ref{fig:colstats}. During the creation of each 
collection an upper limit was defined to limit the size of documents.

\begin{figure}[tb]
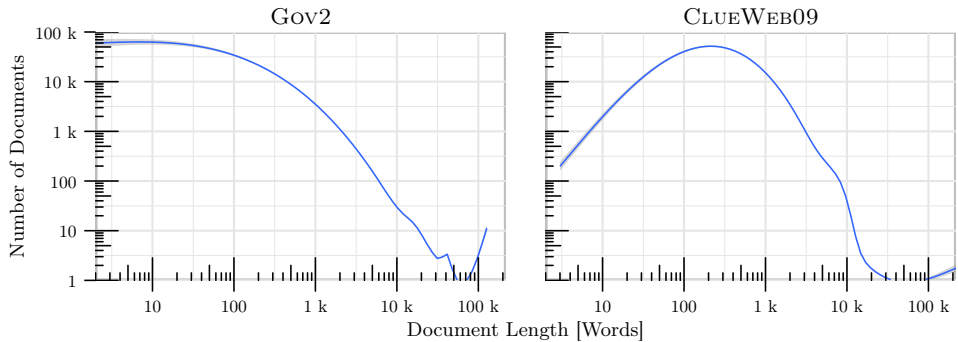


\vspace{-4ex}
\caption{Document lengths distribution for \govtwo\ and \clueweb.}
\label{fig:colstats}
\end{figure}

Figure~\ref{fig:timeclueweb} shows results for \rankand\ and \rankor\ retrieval 
over the larger \clueweb\ collection with similar results to that shown 
for \govtwo\ above.

\begin{figure}
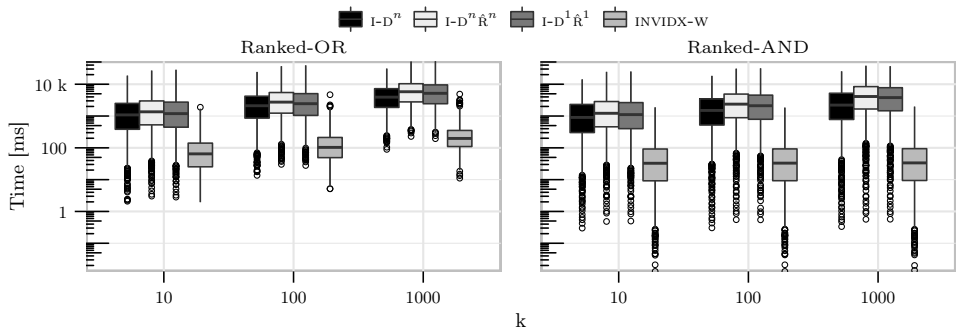

%
\caption{Disjunctive (\rankor) and conjunctive (\rankand) top-$k$ query time for $k=10, 100$ and $1000$ for \clueweb\ and query sets and \bm\ in milliseconds.}
\label{fig:timeclueweb}
\end{figure}

\end{document}